\documentclass[a4paper,UKenglish]{article}
 
\usepackage{microtype}
\usepackage{geometry}
\usepackage[all]{xy} 
\usepackage{marvosym}
\usepackage{amssymb,mathrsfs,stmaryrd,bbm,amsmath,amsthm}
\usepackage{multicol}
\usepackage{nicefrac}
\usepackage{tikz}
\usepackage{enumerate}

\title{The positivication of coalgebraic logics\footnote{This work was supported by the ERC grant ERC-2015-STG, ProFoundNet}}
\date{}

\author{Fredrik Dahlqvist\\University College London\\\texttt{f.dahlqvist@ucl.ac.uk}
\and
Alexander Kurz \\ University of Leicester \\  \texttt{ak155@leicester.ac.uk}}

\newcommand{\Set}{\mathbf{Set}}
\newcommand{\Pos}{\mathbf{Pos}}
\newcommand{\BA}{\mathbf{BA}}
\newcommand{\CABA}{\mathbf{CABA}}
\newcommand{\DL}{\mathbf{DL}}

\newcommand{\Pow}{\mathsf{P}} 
\newcommand{\cPow}{\mathcal{P}} 
\newcommand{\Nb}{\mathsf{N}} 
\newcommand{\MNb}{\mathsf{M}} 
\newcommand{\Mset}{\mathsf{B}} 
\newcommand{\Free}{\mathsf{F}}
\newcommand{\Forg}{\mathsf{U}} 
\newcommand{\Uf}{\mathsf{S}}
\newcommand{\Pf}{\mathsf{S'}}
\newcommand{\Up}{\Pow'}
\newcommand{\ForgBADL}{\mathsf{W}} 
\newcommand{\FreeBADL}{\mathsf{G}} 
\newcommand{\Disc}{\mathsf{D}} 
\newcommand{\Con}{\mathsf{C}} 
\newcommand{\ForgPos}{\mathsf{V}} 
\newcommand{\UP}{\mathsf{Up}}
\newcommand{\DO}{\mathsf{Down}}
\newcommand{\BASet}{\mathsf{U_{BA}}}
\newcommand{\ForgBA}{\mathsf{U_{BA}}}
\newcommand{\CABASet}{\mathsf{U_{CABA}}}
\newcommand{\ForgCABA}{\mathsf{U_{CABA}}}
\newcommand{\FreeBA}{\mathsf{F}_\mathsf{BA}}
\newcommand{\FreeDL}{\mathsf{F}_\mathsf{DL}}
\newcommand{\FreeCABA}{\mathsf{F_{CABA}}}
\newcommand{\poly}[1][\Sigma]{\mathsf{H}_{#1}}
\newcommand{\ana}[1][\Sigma]{\mathsf{G}_{#1}}

\newcommand{\id}{\mathrm{id}}

\newcommand{\Var}{\mathcal{A}}
\newcommand{\two}{\mathbbm{2}}

\newcommand{\op}{^{\mathrm{op}}}

\newcommand{\Ran}{\mathrm{Ran}}
\newcommand{\Lan}{\mathrm{Lan}}
\newcommand{\Alg}{\mathrm{Alg}}
\newcommand{\cop}{\bullet} 

\newcommand{\dia}{\Diamond}
\newcommand{\diak}[1][k]{\dia_k}
\newcommand{\lsem}{\llbracket}
\newcommand{\rsem}{\rrbracket}

\newcommand{\inv}{^{-1}}
\newcommand{\N}{\mathbb{N}} 
\newcommand{\up}{\hspace{1pt}\uparrow\hspace{-2pt}}
\newcommand{\down}{\hspace{1pt}\downarrow\hspace{-2pt}}

\newcommand{\ari}{\mathrm{ar}}

\theoremstyle{plain}
\newtheorem{theorem}{Theorem}
\newtheorem{proposition}[theorem]{Proposition}

\newtheorem{corollary}[theorem]{Corollary}
\newtheorem{lemma}[theorem]{Lemma}
\newtheorem{remark}[theorem]{Remark}
\theoremstyle{definition}
\newtheorem{definition}{Definition}

\newtheorem{rem}[definition]{Remark}

\begin{document}

\maketitle

\begin{abstract}
We present positive coalgebraic logic in full generality, and show how to obtain a positive coalgebraic logic from a boolean one. On the model side this involves canonically computing a endofunctor $T': \Pos\to\Pos$ from an endofunctor $T: \Set\to\Set$, in a procedure previously defined by the second author \textit{et alii} called \emph{posetification}. On the syntax side, it involves canonically computing a syntax-building functor $L': \DL\to\DL$ from a syntax-building functor $L: \BA\to\BA$, in a dual procedure which we call \emph{positivication}. These operations are interesting in their own right and we explicitly compute posetifications and positivications in the case of several modal logics. We show how the semantics of a boolean coalgebraic logic can be canonically lifted to define a semantics for its positive fragment, and that weak completeness transfers from the boolean case to the positive case.
\end{abstract}

\section{Introduction}\label{sec:Intro}

Partially ordered structures are ubiquitous in theoretical computer science. From knowledge representation to abstract interpretation in static analysis, from resource modelling to protocol or access rights formalization in formal security, the list of applications is enormous. Being able to formally reason about transition systems over posets therefore seems important, but has not been systematically developed. The natural formalism to reason about transition systems is undoubtedly the class of \emph{modal} logics. However, most are tailored to transition structures over \emph{sets}. This is a direct consequence of the fact that most modal logics are \emph{boolean}. \emph{Positive modal logic} is the exception, and is most naturally interpreted in partially ordered Kripke structures (see for example \cite{1997:CelaniJansana,2004:gehrkeSahlqvist}).

Arguably, the most natural and powerful framework to study boolean modal logics in a uniform and systematic way, is the theory of \emph{Boolean Coalgebraic Logics} (henceforth BCL, see e.g. \cite{2009:Cirstea:MLC}). In its `abstract' (\cite{2011:DirkOverview}) presentation, it is parametrised by an endofunctor $L:\BA\to\BA$ which builds modal algebras of modal terms over a boolean structure, an endofunctor $T:\Set\to\Set$ which builds the transition structures over which the modal terms are to be interpreted, and a natural transformation $\delta: L\Pow\to \Pow T\op$ (where $\Pow:\Set\op\to\BA$ is the powerset functor) which implements the interpretation by associating sets of acceptable successors states to each modal term over a predicate (see \cite{2004:KKP,2005:KKP,2010:JacobsExemplaric,2012:kurz-rosicky}). This data, and the dual adjunction between $\Set$ and $\BA$, is traditionally summarized in the following diagram
\begin{equation}\label{diag:boolean}
\xymatrix
{
\BA\ar@/^1pc/[rr]^{\Uf} \ar@(dl,ul)^{L}& \perp&\Set\op\ar@/^1pc/[ll]^{\Pow}\ar@(dr,ur)_{T\op} & & & \delta: L\Pow\Rightarrow \Pow T\op
}
\end{equation}
where $\Uf$ is the functor sending a boolean algebra to the set of its ultrafilters.

To develop an equally powerful framework for reasoning about transition structures over posets, it seems natural to study \emph{Positive Coalgebraic Logics} (henceforth PCL). In fact, work in this direction has already started, see for example \cite{2012:PCLExpressivity,2013:PositiveFragments}. We pursue this work further and present PCL in full generality, i.e. at the same level of generality as its boolean counterpart. Moreover, given the close kinship between the two theories, we will show that the wheel needn't be re-invented every time, and that many BCLs have a canonical positive fragment which inherits useful properties from its boolean parent. The data defining a PCL will be
\begin{equation}\label{diag:positive}
\xymatrix
{
\DL\ar@/^1pc/[rr]^{\Pf} \ar@(dl,ul)^{L'}& \perp&\Pos\op\ar@/^1pc/[ll]^{\Up}\ar@(dr,ur)_{(T')\op} & & & \delta': L'\Up\Rightarrow \Up (T')\op
}
\end{equation}
where $\Pf$ is the functor sending a distributive lattice to the poset of its prime filters, and $\Up$ is the functor sending a poset to the distributive lattice of its upsets. The following observation will be of fundamental importance in what follows: the adjunction $\Pf\dashv\Up: \Pos\op\to\DL$ which is the backbone of diagram (\ref{diag:positive}) is in fact \emph{$\Pos$-enriched}; that is to say $\DL$ and $\Pos$ are $\Pos$-enriched categories and $\Pf,\Up$ are $\Pos$-enriched functors (\cite{1982:kelly}). Clearly, it would be a shame not to use this extra structure which comes for free. But more seriously, this enriched structure is not simply a mathematical quirk, it suggests that `doing logic' positively is quite different from `doing logic' in a boolean setting, in particular it is more than simply dropping negations. In fact \emph{inequations} become the standard relation between terms on the syntax side, just as it is between elements on the model side. This is borne out by the existing axiomatization of positive modal logic originally proposed in \cite{1995:DunnPML} which is entirely given by inequations. For these reasons, and following \cite{2012:PCLExpressivity,2013:PositiveFragments}, this paper will present positive coalgebraic logic as a $\Pos$-enriched coalgebraic logic. In a slogan: ``\emph{in positive coalgebraic logic we remove negations but we add order}''.

Working in a $\Pos$-enriched setting means that the syntax-building functor $L':\DL\to\DL$ and the coalgebra-building functor $T':\Pos\to\Pos$ will also need to be $\Pos$-enriched. The first main contributions of this paper is to show how well-known ordinary functors from BCL can be turned into $\Pos$-enriched functors performing analogous roles in PCL. 
On the semantics side this means turning an ordinary functor $T:\Set\to\Set$ into a $\Pos$-enriched functor $T':\Pos\to\Pos$ by a process called \emph{posetification} first developed in \cite{2013:PositiveFragments}, for which we develop a practical understanding in Section \ref{sec:Posetification} by computing the posetification of several well-known functors from modal logic: the neighbourhood, monotone neighbourhood, powerset and multiset functors. On the syntax side this means turning an ordinary functor $L:\BA\to\BA$ into a $\Pos$-enriched functor $L':\DL\to\DL$, a process which we call \emph{positivication} and which is detailed in Section \ref{sec:Positivication}. We show how positivication can be applied to the functor defining normal modal logic, but also to functors which define non-monotone modal logics. In this case, the positivication procedure may not yield a logic at all, at least not in the usual meaning of the word. The second main contribution of this paper is to show how a semantic natural transformation $\delta: L\Pow\to\Pow T$ can also be lifted to define a $\Pos$-enriched semantic natural transformation $\delta': L'\Up\to\Up T'$, where $T'$ is the posetification of $T$, and $L'$ is the positivication of $L$. This is done in Section \ref{sec:PCL}, where we also prove that by construction of $\delta'$, if $\delta$ defines a weakly complete logic, then so does $\delta'$.

\textbf{Related work.} As mentioned above several aspects of positive coalgebraic logics  have been studied in \cite{2012:PCLExpressivity,2013:PositiveFragments}, the concept of posetification and the idea of working in an enriched setting in particular. The key contribution of this work is to dissociate the syntax from the semantics. This reflects the practise of modal logics, where the syntax-building functor $L$ is usually not defined directly from the semantics-building functor, but rather from a grammar which is convenient to express certain properties. Graded modal logic for example relies on a syntax which is not obviously related to its semantics. This justifies going beyond the techniques of \cite{2013:PositiveFragments}. As a consequence, one must also be able to define a semantic natural transformation $L'\Up\to \Up T'$, which we do be adapting the boolean semantics. We are also indebted to work on monotone modal logic for the monotone neighbourhood functor, see eg  \cite{2004:hansenMonotone, 2010:santocanaleUniform},  and on non-monotone modal logic for the (unrestricted) neighbourhood functor as these two cases highlight many of the peculiar features of our approach.

\section{A maths toolkit}\label{sec:maths}

\newcommand{\DLBA}{\mathsf{K}}
\newcommand{\FreeDLBA}{\mathsf{G}}
\newcommand{\BADL}{\ForgBADL}
\newcommand{\Id}{\mathsf{Id}}

\subsection{Ordinary vs $\Pos$-enriched category theory}
The central tool of this paper is to work in categories enriched over $\Pos$. For a general reference to enriched categories we refer to \cite{1982:kelly}. But the special case of $\Pos$-categories is much simpler than the general case and we believe that most of this paper can be read without special knowledge in enriched category theory. The purpose of this section is to review what will be required.

A $\Pos$-category is a category in which homsets are posets and composition is monotone in each argument. A $\Pos$-functor is a functor that is \emph{locally monotone}, that is, it preserves the order on homsets. $\Pos$-natural transformations are just natural transformations.

Monotonicity permeates all aspects of $\Pos$-enriched categories. For example, $\Pos$-enriched algebra, or \emph{ordered algebra}, is characterised by all operations of an ordered algebra being monotone. This is important for our application of ordered algebra to positive coalgebraic logic, that is, to coalgebraic logic with monotone operations.

The basic features of ordered universal algebra can be developed in much the same way as ordinary universal algebra \cite{BloomWright}. Following the $\Pos$-enriched approach of  \cite{2016:ordered-algebras}, the most important change to make is to replace coequalisers by so-called coinserters. 

One of the most important features of the $\Pos$-enriched setting is that with the so-called weighted limits and colimits additional universal constructions become available. For example,
$$
\xymatrix@C=8ex
{
A\ar@<6pt>[r]^{g}\ar@<-6pt>[r]_{f}^{\uparrow} 
& B \ar[r]^{c} 
& C
}$$
is a \emph{coinserter} if $c\circ f\le c\circ g$ and for all $h$ with $h\circ f\le h\circ g$ there is a unique $k$ such that $k\circ c=h$. This is almost like a coequaliser, but $C$ is a quotient of $B$ w.r.t. inequations. For example, in $\Pos$, the coinserter $C$ is obtained by adding to $B$ the inequations $\{ fa\le ga \mid a\in A\}$ and then quotienting by anti-symmetry. We will encounter two special kinds of coinserters (we sometimes drop now the $\uparrow$ notation in the interest of typesetting):
$$
\hfill\xymatrix@C=8ex
{
A\ar@<2pt>[r]^{\pi_1}\ar@<-2pt>[r]_{\pi_0} \ar@(dl,ul)^{s}
& B \ar@/_1.5pc/[l]|-i
}\hfill$$
A pair of arrows, and also its coinserter, is called \emph{reflexive} if there is $i:B\to A$ such that $\pi_0\circ i = \pi_1\circ i = \id$ and it is called \emph{symmetric} if there is $s:A\to A$ such that $\pi_0\circ s=\pi_1$ and $\pi_1\circ s = \pi_0$. Note that an arrow is a coinserter of a symmetric pair, or a symmetric coinserter for short, iff it is a coequaliser of the same pair.

The dual notion is that of an \emph{inserter}
$$
\xymatrix@C=8ex
{
E \ar[r]^e
& B\ar@<6pt>[r]^{g}\ar@<-6pt>[r]_{f}^{\uparrow} 
& A . 
}$$
In $\Pos$, as well as in categories of ordered algebras, the inserter $E$ is $\{b\in B\mid fa\le gb\}$. 

Another important limit not available in ordinary categories is the \emph{power} or \emph{cotensor} with a poset. For example, in $\Pos$ we have $X^\two=\{(x,x')\mid x\le x'\}$ where $\two$ is $\{0<1\}$. We will also encounter the dual notion, the \emph{tensor} or \emph{copower} $X\bullet A$. Here we say that a category $\cal A$ has tensors if for all $A\in\cal A$ and all posets $X$, there is an object $X\bullet A$ such that 
$${\cal A}(X\bullet A,A')\cong[X,{\cal A}(A,A')]$$
 where $[X,Y]$ denotes exponentiation (aka internal hom) in $\Pos$. The tensor $\two\bullet A$ can be understood as an ordered coproduct of $A$ with itself in which ``each $a$ on the left is smaller than the $a$ on the right''.

In order to treat ordinary categories and $\Pos$-enriched categories in the same framework, we consider an ordinary category as a $\Pos$-category with discrete homsets. For each $\Pos$-category $\cal A$ there is a corresponding ordinary category ${\cal A}_o$. For example, we have $\Set=\Set_o$ and $\BA=\BA_o$, but $\Pos$ and $\Pos_o$ are different. In particular, there is no (enriched) forgetful functor $\Pos\to\Set$, only an (ordinary) forgetful functor $\Pos_o\to\Set_o$. Note that using ``$_o$'' allows us to drop the qualifications enriched and ordinary without creating ambiguity. For example, the inclusion $\Disc:\Set\to\Pos$ has a  left adjoint $\Con:\Pos\to\Set$ mapping a poset to its connected components
 and the inclusion $\Disc_o:\Set_o\to\Pos_o$ has as a right adjoint the forgetful functor $\ForgPos:\Pos_o\to\Set_o$, but $\Disc:\Set\to\Pos$ does not have a right adjoint.%
\footnote{$L:\cal A \to\cal B$ is a $\Pos$-enriched left-adjoint of $R:\cal B\to\cal A$ if there is a natural isomorphism \emph{of posets} ${\cal B}(LA,B)\cong{\cal A}(A,RB)$. We have $\Pos_o(DA,B)\cong\Set_o(A,\ForgPos B)$ but not $\Pos(DA,B)\cong\Set(A,\ForgPos B)$. }
\begin{equation}\label{eq:CDV}
\hfill 
\Con\dashv \Disc : \Set\to\Pos \quad\quad\quad\quad \Disc_o\dashv \ForgPos:\Pos_o\to\Set_o \hfill
\end{equation}
Accordingly, $\Disc$ preserves all (weighted) limits and $\Disc_o$ preserves all (ordinary) colimits. But $\Disc$ does not preserve all ($\Pos$-enriched) colimits and indeed we will see later that $\Disc$ does not preserve all coinserters.

We will also need the corresponding results on the algebraic side. The inclusion $\ForgBADL:\BA\to\DL$ has a right adjoint $\DLBA$ (mapping a $\DL$ to the largest Boolean subalgebra it contains) and $\ForgBADL_o:\BA_o\to\DL_o$ has a left-adjoint $G:\DL_o\to\BA_o$ (mapping a distributive lattice to the free $\BA$ over it).
\begin{equation}\label{eq:GWK}
\hfill
\DLBA\vdash \BADL: \BA\to\DL
\quad\quad\quad\quad 
\BADL_o\vdash \FreeDLBA : \DL_o\to\BA_o 
\hfill
\end{equation}
Note that \eqref{eq:CDV} and \eqref{eq:GWK} are dually equivalent when restricted to finite structures.

\subsection{The ordered variety of Boolean algebras}\label{sec:BAvariety}

The category $\BA$ of Boolean algebras has discrete homsets, giving rise to a forgetful functor $\BA\to\Set$. This functor is a $\Pos$-enriched, or ordered, variety \cite{2016:ordered-algebras}. At the heart of this observation is the fact that reflexive coinserters in $\BA$ are symmetric, a fact that we prove in detail as it will be important later. 

\begin{proposition}\label{prop:symmetric}
Every reflexive pair
$
\xymatrix@C=8ex
{
A_1\ar@<2pt>[r]^{\pi_1}\ar@<-2pt>[r]_{\pi_0} & A_0 \ar@/_1.5pc/[l]|-i
}$
in $\BA$ is symmetric.
\end{proposition}

\begin{proof}
Let $(a,b)\in A_1$. Here, $(a,b)$ is a shorthand for an element in $A_1$ such that $\pi_0((a,b))=a$ and $\pi_1((a,b))=b$.
We write $(a,a)$ for $i(a)$.
Consider the Boolean algebra morphism $\phi: A_1\to A_1$ defined by
\begin{align*}
\phi(a,b)=\Big((a,b) \wedge (b,b)\to (a,a)\Big) \wedge 
\Big((a,b) \wedge (a,a)\to (b,b)\Big) \wedge
\Big((a,a) \vee (a,b) \vee (b,b)\Big)
\end{align*}
\noindent Since the projections are $\BA$-morphisms, we obtain
$\pi_0(\phi(a,b)) =  (a\to b) \wedge (a\vee b) = b$ and 
$\pi_1(\phi(a,b)) =  (b\to a) \wedge (a\vee b) = a$, showing $(b,a)\in A_1$.
\end{proof}
(The argument also works for Heyting algebras.) Note that we can equip $\BA$ with a forgetful functor $\BA\stackrel{\BADL}{\longrightarrow}\DL\to\Pos$ mapping each BA to its carrier in its natural order, but this functor is not an ordered variety since $\BA$ is not closed under weighted limits in $\DL$.

\subsection{Density and Kan extensions}
Much of our technical work revolves around the result that $\Disc:\Set\to\Pos$ is dense, see \cite{2013:PositiveFragments}, and that $\BADL:\BA\to\DL$ is codense, see Theorem~\ref{thm:BADLdense}. This in turn allows us to extend functors on $\Set$ to functors on $\Pos$ via left Kan extension and to extend functors on $\BA$ to functors on $\DL$ via right Kan extensions, as we will review now.

A functor $K:\cal A\to \cal C$ is \emph{dense} if colimit preserving functors $\cal C\to\cal B$ are determined by their restriction along $K$, or, more formally, if the functor $[K,\Id_{\cal B}]:[\cal C, \cal B]\to[\cal A,\cal B]$ restricting along $K$ is fully faithful \cite[Thm 5.1]{1982:kelly}. If, moreover, $K$ itself is fully faithful, then a colimit preserving functor $\cal C\to\cal B$ is the left Kan extension of its restriction along $K$ \cite[Thm 5.29]{1982:kelly}. Furthermore, we may be able to compute left Kan extensions explicitely with the help of a so-called density presentation \cite[Thm 5.19]{1982:kelly}. For example, we know (see \cite{2013:PositiveFragments}) that $\Disc:\Set\to\Pos$ is dense and has a density presentation given by reflexive coinserters of `nerves of posets'. Explicitly, every poset $X$ is the reflexive coinserter 
\footnote{The coinserter of $
\xymatrix@C=8ex
{
VX^\two\ar@<1pt>[r]^{}\ar@<-1pt>[r]_{} &  VX
}
$
in $\Set$ provides an example of weighted colimit that is not preserved by $\Disc$, showing that $\Disc$ cannot have an (enriched) right adjoint.}
\begin{equation}\label{eq:denspresD}
\hfill \xymatrix@C=8ex
{
\Disc\ForgPos X^\two\ar@<2pt>[r]^{\Disc \pi_1}\ar@<-2pt>[r]_{\Disc \pi_0} & \Disc \ForgPos X\ar[r]^{}\ar@/_1.5pc/[l]|-i & X
}\hfill 
\end{equation}
where $\ForgPos X^\two=\{(x,x')\mid x\le_X x'\}$. %
That the coinserter is reflexive means that $\Disc\pi_0\circ i = \Disc\pi_1\circ i = \id$, which is true for $i(x)=(x,x)$.  The fact that these coinserters provide a density presentation means that the left Kan extension of a functor $F:\Set\to\Pos$ along $\Disc$ can be computed as the coinserter
\begin{equation}\label{eq:leftKanalongD}
\hfill\xymatrix@C=8ex
{
F \ForgPos X^\two\ar@<2pt>[r]^{F \pi_1}\ar@<-2pt>[r]_{F \pi_0} & F \ForgPos X\ar[r]^{}
& (\Lan_\Disc F)X
}\hfill
\end{equation}
and we will see examples of this in the next section. 
If one happens to extend along an adjoint functor $K$, Kan-extensions are easier:
\begin{equation}\label{eq:Kanalongadjoint}
\hfill
K\dashv V \ \Longrightarrow\  \Lan_KF=FV 
\hfill
G\dashv K \ \Longrightarrow\   \Ran_KF=FG
\hfill
\end{equation}
This implies that to compute (ordinary) Kan extensions along $\Disc_o$ or $\BADL_o$ we can use \eqref{eq:CDV} or \eqref{eq:GWK}, and \eqref{eq:Kanalongadjoint}. To better understand the difference with extending along $\Disc$ or $\BADL$, we can see the computation of the Kan extensions in two steps:
\begin{equation}\label{eq:2stageKan}
\xymatrix@R=5ex{
{\cal C} \ar@{->}[r]^{H'} & {\cal C} \\
{\cal C}_o \ar[u]_{} \ar@{->}[r]^{\tilde H} & {\cal C}_o \ar[u]^{} \\
{\cal A}_o\ar[u]_{{J_o}}  \ar@/^20pt/[uu]|-{J}  \ar[r]_{H} 
& {\cal A}_o \ar[u]^{{J_o}}\ar@/_20pt/[uu]|-{J}
}
\hfill\hspace{3em}
\xymatrix@R=5ex{
\Pos \ar@{->}[r]^{H'} & \Pos \\
\Pos_o \ar[u]^{} \ar@{->}[r]^{\tilde H} & \Pos_o\ar[u]_{} 
\\
\Set_o
\ar@/^8pt/[u]|-{\Disc_o}_{\;\dashv} \ar@/_8pt/@{<-}[u]|-{\ForgPos} 
\ar[r]_{H} \ar@/^20pt/[uu]|-{\Disc}
& \Set_o \ar[u]^{\Disc_o} \ar@/_20pt/[uu]|-{\Disc}
}
\hfill\hspace{3em}
\xymatrix@R=5ex{
\DL \ar@{->}[r]^{H'} & \DL \\
\DL_o \ar[u]^{} \ar@{->}[r]^{\tilde H} & \DL_o\ar[u]_{} 
\\
\BA_o
\ar@/^8pt/[u]|-{W_o}_{\ \vdash} \ar@/_8pt/@{<-}[u]|-{\FreeDLBA} 
\ar[r]_{H} \ar@/^20pt/[uu]|-{\BADL}
& \BA_o \ar[u]^{W_o} \ar@/_20pt/[uu]|-{\BADL}
}
\hfill
\end{equation}
(i) Since ${\cal C}_o\to{\cal C}$ preserves all ordinary (co)limits we can use \cite[Thm 4.47]{1982:kelly} to break down the extension of $JH$ along $J$ into first extending $J_oH$ along $J_o$ to $\tilde H$ and then extending $\tilde H$ to $H'$. 
(ii) If a functor ${\cal C}_o\to{\cal C}_o$ is locally monotone then this functor is its own extension (both left and right) to ${\cal C}\to{\cal C}$. This means that for locally monotone functors $\tilde H$ the upper square is trivial. 
(iii) To compute $\tilde H$ we can use  \eqref{eq:CDV} or \eqref{eq:GWK}, and \eqref{eq:Kanalongadjoint}. --- While this is sometimes a good approach, the downside is that $\tilde H$ is typically not locally monotone (so we cannot use (i)) and the inclusion ${\cal C}_o\to{\cal C}$ is not fully faithful (so we loose  the good properties of Kan extensions along fully faithful functors). To summarise, to compute Kan extensions along $\Disc$ we use \eqref{eq:leftKanalongD} and for $\BADL$ we will develop a similar presentation in Theorem~\ref{thm:BADLdense}.

\section{Posetification}\label{sec:Posetification}
\newcommand{\EM}{\trianglelefteq}
\newcommand{\Conv}{\mathrm{Conv}}

We define the \emph{posetification $T'$ of a $\Set$-functor $T$} as $\Lan_\Disc \Disc T$, the left Kan-extension of $\Disc T$ along $\Disc$, together with a natural isomorphism $\alpha:\Disc T\Rightarrow T'\Disc$. In concrete examples we will typically define $T'$ so that $\alpha$ is the identity.
In particular, $T'$ is the universal locally monotone extension of $T$, that is, for all $S:\Pos\to\Pos$ and $\beta:\Disc T\Rightarrow S\Disc$, there is a unique $\gamma:T'\to S$ such that $\gamma\circ\alpha=\beta$. The coinserter (\ref{eq:leftKanalongD}) now becomes
\begin{equation}\label{eq:posetification}
\hfill\xymatrix@C=8ex
{
\Disc T \ForgPos X^\two\ar@<6pt>[r]^{\Disc T \pi_1}\ar@<1pt>[r]_{\Disc T \pi_0} & \Disc T \ForgPos X\ar[r]^{e_X}& T' X.
}\hfill
\end{equation}
It is computed for any poset $X$ in the following way:
\begin{enumerate}[(i)]
\item consider the (reflexive) relation $R_T\subseteq T\ForgPos X\times T\ForgPos X $ given by
\[
(a,b)\in R_T \Leftrightarrow \exists c\in T\ForgPos X^\two . \hspace{1ex}T\pi_0(c) = a \hspace{1ex}\& \hspace{1ex} T\pi_1(c)=b
\]
\item compute its transitive closure $\leq_T$
\item quotient $T\ForgPos X$ by the equivalence relation $\equiv_T\hspace{3pt} =\hspace{3pt} \leq_T\hspace{2pt}\cap\hspace{2pt} \geq_T$,
\item the coinserter is given by  $(T\ForgPos X/\equiv_T,\leq_T)$.
\end{enumerate}

\newcommand{\lift}[1][T]{\overline{#1}\hspace{-3pt}}

Note that by definition, $\ForgPos X^\two\subseteq \ForgPos X\times \ForgPos X$ is precisely the graph of the partial order on $X$. It follows that $R_T$ is simply the \emph{lifting of the partial order on $X$ by the functor $T$} (see \cite[Remark 4.8]{2013:PositiveFragments}), often denoted $\lift\leq$. In the rest of this section we will see examples where $R_T$ is transitive, where $R_T$ is transitive and antisymmetric, and where $R_T$ is not even transitive.


\subsection{Posetification of the covariant powerset functor $\cPow$.}

We recall this case here from  \cite{2013:PositiveFragments}  because it illustrates the steps (ii)-(iv) of the posetification procedure very clearly. We start by defining the relation $R_{\cPow}\subseteq \cPow \ForgPos X\times \cPow \ForgPos X$ by
\[
\hfill (a,b)\in R_{\cPow} \Leftrightarrow \exists c\in\cPow \ForgPos X^\two .\, \pi_0[c]=a \ \& \ \pi_1[c]=b\hfill
\]
which means that
\[\hfill
(a,b)\in R_{\cPow} \ \Longleftrightarrow \ (\forall x\in a)( \exists y\in b)\hspace{2pt} .\hspace{2pt} x\leq y\text{ and }(\forall y\in b) (\exists x\in a)\hspace{2pt} .\hspace{2pt}  x\leq y
\hfill\]
It is well-known that $\cPow$ preserves weak-pullbacks, and that this guarantees the transitivity of $R_{\cPow}$. We can thus skip step (i) of the posetification procedure since $R_{\cPow}=\hspace{2pt}\leq_{\cPow}$. 
The relation $R_{\cPow}=\lift[\cPow]\leq$ is known as the Egli-Milner (pre-)order associated with $\leq$. It is not hard to check that $a\equiv_{\cPow} b$, i.e. $a \leq_{\cPow} b$ and $b \leq_{\cPow} a$, iff $\Conv(a)=\Conv(b)$ where $\Conv(a)$ is the convex closure of $a$, i.e. the set $\{x\in X\mid \exists y_1,y_2\in a, y_1\leq x\leq y_2\}$, and that $a\equiv_\cPow \Conv(a)$. It follows that $\cPow' X=(\cPow \ForgPos X/\equiv_{\cPow},R_{\cPow})=(\{\Conv(a)\mid a\subseteq X\},R_{\cPow})$.

\subsection{Posetification of analytic functors}
Consider first a polynomial functor $\poly:\Set\to\Set$ given by a signature $\ari:\Sigma\to\N$ and a collection of set $(A_\sigma)_{\sigma\in\Sigma}$
\[
\begin{cases}
\poly  X=\coprod_{\sigma\in \Sigma}A_\sigma\times X^{\ari(\sigma)}\\
\poly f=\coprod_{\sigma\in \Sigma} \id_{A_\sigma}\times f^{\ari(\sigma)}
\end{cases}
\]
By definition of $\poly$ on morphisms, the relation $R_{\poly[]}\subseteq \poly \ForgPos X\times \poly \ForgPos X$ is given by
\[
\big((a,x_1,\ldots, x_{\ari(\sigma)}),(b,y_1,\ldots, y_{\ari(\sigma')})\big)\in R_{\poly[]} \Leftrightarrow a=b, \sigma=\sigma', x_i\leq y_i, 1\leq i\leq \ari(\sigma)
\]
Since polynomial functors preserve weak-pullbacks we have $R_{\poly[]}=\lift[\poly]\leq \hspace{2pt}=\hspace{2pt}\leq_{\poly[]}$. Moreover, it is easy to see from the definition above that $\leq_{\poly[]}$ is anti-symmetric (since $\leq$ is). It follows that $\poly' X$ is simply given by $(\poly \ForgPos X,R_{\poly[]})$. 

\noindent We can now compute the posetification of \emph{analytic functors} (\cite{1986:JoyalFoncteursAnal}), i.e. functors of the shape
\[
\ana  X=\coprod_{\sigma\in \Sigma}A_\sigma\times (X^{\ari(\sigma)}/G_\sigma)
\]
where each quotient $X^{\ari(\sigma)}/G_\sigma$ is taken with respect to the obvious action of a subgroup of the permutation group $G_\sigma\subseteq \mathrm{Perm}(\ari(\sigma))$ on the tuples of $X^{\ari(\sigma)}$. The most well-known example is the `bag' or `multiset' functor which is given by the choice $\Sigma=\N$, $\ari=\id_{\N}$ and $G_n=\mathrm{Perm}(n), n\in \N$. Analytical functors preserve weak-pullbacks (in fact wide pullbacks, see \cite{2011:AdamekPresentation}), and thus $R_{\ana[]}=\leq _{\ana[]}$.


\begin{proposition}
The posetification of an analytic functor $\ana:\Set\to\Set$ is given by $\ana' X=(\ana \ForgPos X, \lift[\ana]\leq)$.
\end{proposition}
\begin{proof}
To simplify the notation we assume that each $A_\sigma=1$, fix an element $\sigma$ with arity $\ari(\sigma)=n$, and denote by $[(x_1,\ldots,x_n)]$ the equivalence class of the tuple $(x_1,\ldots,x_n)$ under the action of $G_\sigma$. Note that by definition of $\ana$, two elements of $\ana \ForgPos X$ can only be related by $\leq_{\ana[]}$ if they belong to the same $\sigma$-component of the coproduct. Moreover, if $[(x_1,\ldots,x_n)]\leq_{\ana[]} [(y_1,\ldots,y_n)]$, then by definition there exists a permutation $\pi\in G_\sigma$ such that $(x_1,\ldots,x_n) \leq (y_{\pi(1)},\ldots,y_{\pi(n)})$ (where $\leq$ here is component-pointwise). Similarly, if $[(y_1,\ldots,y_n)]\leq_{\ana[]} [(x_1,\ldots,x_n)]$, there exists a permutation $\rho\in G_\sigma$ such that $(y_1,\ldots,y_n) \leq (x_{\rho(1)},\ldots,x_{\rho(n)})$. It follows that $(x_1,\ldots,x_n) \leq (x_{\pi(\rho(1))},\ldots,x_{\pi(\rho(n))})$, and since $\pi\rho$ is of finite order we easily get by iterating at most $n$ times that $(x_1,\ldots,x_n) = (x_{\pi(\rho(1))},\ldots,x_{\pi(\rho(n))})$. This in turn implies that $(y_{\pi(1)},\ldots,y_{\pi(n)})\leq (x_1,\ldots,x_n)$, from which we can conclude that $[(x_1,\ldots,x_n)]=[(y_1,\ldots,y_n)]$. Thus $\leq_{\ana[]}=\lift[\ana]\leq$ is anti-symmetric.
 
%
\end{proof}
In particular the posetification of the bag functor $\Mset$ is given by $\Mset' (X,\leq)=(\Mset \ForgPos X, \lift[\Mset]\leq)$.


\subsection{Posetification of the monotone neighbourhood functor $\MNb$} 

Recall that the \emph{monotone neighbourhood functor} $\MNb:\Set\to\Set$ if defined on sets by 
\[\hfill
\MNb X=\{A\subseteq \cPow X\mid U\in A, U\subseteq V\Rightarrow V\in A\}\hfill
\]
and on functions $f:X\to Y$ by taking the double inverse image $(f\inv)\inv: \MNb X\to\MNb Y$. It is not hard to check that $\MNb f$ can be described  more simply as $\MNb f(A)=\up f[A]$, where $\up f[A]$ is the upward closure (under inclusion) of the direct image of $A$ by $f$. With this in place we can compute the coinserter (\ref{eq:posetification})
\[\hfill
\xymatrix@C=8ex
{
\Disc \MNb \ForgPos X^\two\ar@<6pt>[r]^{\uparrow \pi_1[-]}\ar@<1pt>[r]_{\uparrow\pi_0[-]} & \Disc \MNb \ForgPos X\ar[r]^{e_X}& \MNb' X.
}\hfill
\]
This time we need to consider the relation $R_{\MNb}\subseteq \MNb \ForgPos X\times \MNb \ForgPos X$ defined by
\[\hfill
(A,B)\in R_{\MNb}\Leftrightarrow\exists C\in \MNb \ForgPos X^\two . \up \pi_0[C]=A \hspace{4pt} \& \hspace{1pt} \up \pi_1[C]=B\hfill
\]
It is known (see \cite{2004:hansenMonotone}) that $\MNb$ does \emph{not} preserve weak pullbacks, and in particular we cannot assume that the relation $R_M$ is transitive. The proof of the following result can essentially be found in Theorem 8.25 of \cite{2003:HelleMastersThesis}\footnote{We thank Clemens Kupke for pointing out this reference.}
\begin{proposition}\label{prop:MNb}
$(A,B)\in {\leq_{\MNb}}$ iff $\forall a\in A.\exists b\in B.\up b\subseteq \up a$ and $\forall b\in B.\exists a\in A.\down a\subseteq \down b$.
\end{proposition}


For any $A\in\MNb \ForgPos X$ and $a\in A$, we write $\up a$ for the upward closure of $a$ under the order of $X$, and $\down(\up (A))$ for the set $\down\{\up a\mid a\in A\}$, where the downward closure is taken with respect to the inclusion. The following corollaries are then easy to check.

\begin{corollary}
$A\equiv_\MNb B$ iff $\down(\up(A))=\down(\up(B))$, and moreover $A\equiv_\MNb \down(\up(A))$.
\end{corollary}

\begin{corollary}
The posetification $\MNb$ is given by $\MNb' X=(\DO(\UP(X)),\leq_{\MNb})$.
\end{corollary}

\subsection{Posetification of the neighbourhood functor $\Nb$.} 

Consider the adjunctions $\FreeBA\dashv\BASet:\BA\to\Set$ and $\FreeCABA\dashv\CABASet:\CABA\to\Set$. 
The monad $\ForgBA\FreeBA$ is naturally isomorphic to the finitary version $\Nb_f$ of the neighbourhood functor $\Nb$. A natural isomorphism is given by the natural transformation $\alpha: \Nb_f\to \Forg_\BA\Free_\BA$ given at each $X$ by
\[\hfill
\alpha_X(A)=\bigvee\{\bigwedge a\wedge \bigwedge (a)^c\mid a\in A\}\hfill
\]
which is indeed a boolean term since $A$ and each $a\in A$ are finite. The inverse of $\alpha$ is built as follows: given a boolean term over $X$ in conjunctive normal form, check for each clause $\bigwedge_{p\in a_1}p\wedge \bigwedge_{q\in a_2}\neg q$ if $a_1\cup a_2=X$, if not rewrite the clause as the equivalent CNF expression
\[\hfill
\bigvee\{\bigwedge_{p\in a_1\cup a_3}\hspace{-2ex}p\hspace{1ex}\wedge \hspace{-2ex}\bigwedge_{q\in a_2\cup (a_3)^c} \hspace{-3ex}\neg q\mid a_3\subseteq X\setminus (a_1\cup a_2)\}\hfill
\]
This yields a finite disjunction $\bigvee\{\bigwedge_{p\in a_i} p\wedge\bigwedge_{q\notin a_i}\neg q\mid i\in I\}$, which we associate with $\{a_i\}_{i\in I}\in \Nb_f X$.
Similarly, the monad $\ForgCABA\FreeCABA$ is naturally isomorphic to the full neighbourhood functor $\Nb$. We can use the special properties of the adjunctions above to compute the posetification of $\Nb_f$ and $\Nb$ indirectly, but relatively straightforwardly. 

\begin{proposition}\label{prop:FreeCoinserter}
Let $X$ be a poset presented by the coinserter $\Disc \ForgPos X^\two \rightrightarrows \Disc \ForgPos X \stackrel{c'}{\longrightarrow} X$ and let $\Free\dashv\Forg$ denote either of the adjunctions $\FreeBA\dashv\ForgBA$ or $\FreeCABA\dashv\ForgCABA$, then the coinserter of $\Free \ForgPos X^\two \rightrightarrows \Free \ForgPos X$ is $\Free c: \Free \ForgPos X\to\Free \Con X$, where $c:\ForgPos X\to \Con X$ is the adjoint of $c'$.
\end{proposition}
\begin{proof}
Note first that $c$ coequalizes $\pi_0,\pi_1:\ForgPos X^\two\rightrightarrows \ForgPos X$; indeed two elements $x,y$ lie in the same connected component precisely when $x\leq y$ or $y\leq x$. It follows that $Fc$ coequalizes $\Free  \ForgPos X^\two \rightrightarrows \Free \ForgPos X$, and in particular $Fc\circ \pi_0\leq Fc\circ \pi_1$. We need to show that it is in fact a coinserter for which, due to Prop.\ref{prop:symmetric}, it is enough to show that it is a coequaliser.
Let $d:\Free\ForgPos X\to Y$ with $d\circ\Free\pi_0 = d\circ \Free\pi_1$.  Let $d':\ForgPos X\to \Forg Y$ be the adjoint transpose of $d$.  We have that $d'$ factors through $c$. Writing $\eta:\Id\to\Forg\Free$ for the unit of the adjunction, it follows that there is a unique $f:\Free\Con X\to Y$ such that $\Forg f\circ\eta_{CX}\circ c= d'$, or, equivalently, that  $f\circ\Free c=d$. We have shown that $\Free c$ is the coequaliser (and coinserter) of $\Free \ForgPos X^\two \rightrightarrows \Free \ForgPos X$.
\end{proof}

\begin{lemma} $\ForgBA$ and $\ForgCABA$ preserve reflexive coequalisers.
\end{lemma}

\begin{proof}
Being a variety of finitary algebras, $\ForgBA$ preserves sifted colimits and, in particular, reflexive coequalisers \cite{2010:adamekAlgebraic}. In the case of $\ForgCABA$ we use that $\CABA$ is equivalent to $\Set\op$ and that $[-,2]:\Set\op\to\Set$ preserves reflexive coequalisers, see \cite[5.1.5 Lemma]{1985:barrWellsTTT}.
\end{proof}

\begin{theorem}\label{thm:NbPoset}
The posetification of $\Nb_f$ is $\Disc\Nb_f\Con$ and the posetification of $\Nb$ is $\Disc\Nb\Con$.
\end{theorem}
\begin{proof}
We use the notation of Proposition~\ref{prop:FreeCoinserter}. It follows from \eqref{eq:CDV} that $\Disc$ preserves all ordinary colimits and, in particular, reflexive coequalizers. Due to the lemma $\Forg$ preserves reflexive coequalizers. Like all functors,  $\Disc$ and $\Forg$ preserve the symmetry of coinserters. It follows from Proposition \ref{prop:FreeCoinserter} that
\[
\hfill
\xymatrix@C=14ex
{
\Disc\Forg\Free\ForgPos X^\two\ar@<-3pt>[r]_{\Disc\Forg\Free \pi_0}\ar@<3pt>[r]^{\Disc\Forg\Free\pi_1} & \Disc\Forg\Free\ForgPos X^\two\ar[r] & \Disc\Forg\Free\Con X
}
\hfill
\]
is a coequalizer and coinserter. Thus $\Disc\Forg\Free\Con$ is the posetification of $\Forg\Free$.
\end{proof}

\begin{rem}
This result is curious at first sight. Due to  \eqref{eq:CDV} and \eqref{eq:Kanalongadjoint}, $\Disc\Nb_f\Con$ and $\Disc\Nb\Con$ are also the \emph{right} Kan-extensions of $\Nb_f$ and $\Nb$, respectively. To better understand the situation let us recall that we need the posetification to be locally monotone, which means that it must be an \emph{enriched} left Kan extension. Now, working in the ordered setting (ie $\Pos$-enriched), Prop.\ref{prop:symmetric} enforces that Boolean algebras cannot be quotiented by a partial order in $\BA$ without quotienting by its symmetric closure. 

For example, let $X=\{p<q\}$, so that $\ForgPos X=\{p,q\}$ and $\ForgPos X^\two=\{(p,p),(p,q),(q,q)\}$. Then dividing $\Free\ForgPos X$ by $p\le q$ `equationally' gives the Boolean algebra $2^3$ whereas the coinserter gives $\Free 1=2^2$.
In more detail: Dividing $\Free\ForgPos X$ by $p\le q$ (or $p\wedge q=p$ or $p\wedge\neg q =0$) gives the Boolean algebra $2^3$ because $\Free\ForgPos X=2^4$ and $p\le q$ kills one of the 4 atoms, namely $p\wedge\neg q$. 
On the other hand, the coinserter divides $\Free\ForgPos X$ by a larger theory, namely by one in which negation is monotone. (Recall that in the $\Pos$-enriched setting all operations are monotone. Of course, one can still have algebras with ``non-monotone'' operations like negation in $\BA$, but then the $\Pos$-enriched order must be discrete. Which does not prevent us from recovering the natural order of $\BA$s by considering $\BA$ as a subcategory of $\DL$.)
\end{rem}

\section{Positivication}\label{sec:Positivication}


As mentioned in the introduction, a boolean coalgebraic logic is given by an endofunctor $T:\Set\to\Set$ determining the type of coalgebraic semantics, and an endofunctor $L:\BA\to\BA$ constructing general `modal algebras'. We have just seen how to extend a functor $T:\Set\to\Set$ to a functor $T':\Pos\to\Pos$ to extend a type of coalgebraic semantics to posets. Now we want to extend a boolean syntax-building functor to a positive syntax-building functor $L':\DL\to\DL$.

An obvious idea is to work dually to the posetification procedure and define the \emph{positivication} of $L:\BA\to\BA$ as the right Kan extension $L'=\Ran_\ForgBADL(\ForgBADL L) $,
where $\ForgBADL$ is the inclusion $\BA\to\DL$.%
\footnote{$\BADL$ is fully faithful. Moreover, whereas $B\in\BA$ is discrete (see Section~\ref{sec:BAvariety}),  $\BADL B\in\DL$ is equipped with its natural order. So while $\BADL$ `forgets negation' it also `adds the order'.}
Note that this gives us what we would expect: (i) an isomorphism $\beta: L'\BADL\cong \BADL L$, saying that $L'$ is the same as $L$ on boolean algebras and (ii) for all $H:\DL\to\DL$ and all $\alpha:H\BADL\to \BADL L$ there is a unique $\gamma:H\to L'$ such that $\beta\circ \gamma = \alpha$, saying that $L'$ is the optimal (or co-universal) extension with (i).

It is also worth emphasising that $\beta: L'\BADL\cong \BADL L$ will be doing some real work once the abstract framework is instantiated with concrete examples. In particular, $\beta^{-1}$ will translate a boolean formula $\phi$ in $LB$ into a positive formula $\beta_B^{-1}(\phi)$ where negation is eliminated from the modal part and pushed ``onto the atoms in $B$''.

In order to capture this process of eliminating negation in the abstract categorical framework, we need to understand, once again, the Kan extension  in the $\Pos$-enriched way. To compute these right Kan extensions we use a presentation of distributive lattices which will play the same role in the computation of positivications as  (\ref{eq:leftKanalongD}) played in the case of posetifications.

\begin{proposition}\label{prop:DLinserter}
Every $A\in\DL$ is the inserter of a diagram of boolean algebras (where $in_1,in_2$ are the canonical embeddings and $e$ is the unit at $A$ of the adjunction $\FreeBADL\dashv \BADL_o$):
\begin{equation}\label{eq:DLinserter}
\hfill
\xymatrix@C=10ex
{
A\ar[r]^e & \ForgBADL \FreeBADL A\ar@<6pt>[r]^-{\ForgBADL \FreeBADL in_2} \ar@<-6pt>[r]_-{\ForgBADL \FreeBADL in_1}^-{\uparrow} & \ForgBADL \FreeBADL (\two\cop A )
}
\hfill
\end{equation}
\end{proposition}

\begin{rem}\label{rem:DLcoinserter}
\begin{enumerate}
\item  The tensor $\two\bullet A$ is isomorphic to $A+A$ modulo inequations $in_1 a\le in_2 a$. If $a\in A$ has a complement then $in_1 a= in_2 a$. If all elements of $A$ are complemented, that is, if  the distributive lattice $A$ happens to be a boolean algebra, then $\two\bullet A\cong A$.
\item
Equivalently, $\two\bullet A$ can be represented as the distributive lattice generated by $\{\Box_1 a\mid a\in A\}$ and $\{\Box_2 a\mid a\in A\}$ modulo equations specifying that $\Box_1, \Box_2$ preserve all $\DL$-operations and modulo inequations $\Box_1 a\le \Box_2 a$.
\item \label{rem:DLcoinserter:duality} Let $A^\partial$ be the Priestley space dual to $A\in\DL$, that is, the space of prime filters on $A$. Then $2\bullet A$ is dual to $(A^\partial)^\two$. \footnote{Cotensors in Priestley spaces are computed as cotensors in $\Pos$, since the forgetful functor from Priestely spaces to posets preserves and creates all $\Pos$-enriched limits.}
\item The inserters \eqref{eq:DLinserter} are reflexive. This follows easily from the definition of tensor with $\two$ as ${\cal A}(\two\bullet A,A')\cong[\two,{\cal A}(A,A')]$ giving us a half-inverse $\two\bullet A\to A$ of both $in_{1},in_{2}:A\to\two\bullet A$ as the transpose of the map $\two\to{\cal A}(A,A)$ which maps both truth values to $\id_A$.
\end{enumerate}
\end{rem}
\newcommand{\cPosDL}{\Up}
\newcommand{\cDLPos}{\Pf}
\newcommand{\cBASet}{\Uf}
\newcommand{\cSetBA}{\Pow}
\newcommand{\SetBA}{\Disc}
\newcommand{\SetPos}{\Disc}
\newcommand{\PosSet}{\ForgPos}
\begin{proof}[Proof of Prop.\ref{prop:DLinserter}]
Let $(-)^\partial$ be the functor that dualises $\DL$s to Priestley spaces. Its composition with the forgetful functor to posets we denote by  $\cDLPos:\DL\to\Pos$. Applying it to \eqref{eq:DLinserter} yields  a reflexive coinserter in $\Pos$
\begin{equation}\label{eq:DLinserter2}
\hfill
\xymatrix@C=10ex
{
\cDLPos A\ar@{<-}[r] & \SetPos \PosSet (\cDLPos A)\ar@{<-}@<6pt>[r]^-{} \ar@{<-}@<-6pt>[r]_-{}^-{\uparrow} & \SetPos \PosSet ( \cDLPos A )^\two
}
\hfill
\end{equation}
as in \eqref{eq:denspresD}.
 Now it only remains to check that it is also a coinserter in Priestley spaces, from which the result follows by duality.
\end{proof}

The following theorem requires some knowledge of enriched category theory. Even though the theorem is one of the main contributions, we encourage the reader so inclined to skip directly ahead to its corollary, which is all that is needed to follow the rest of the paper.

\begin{theorem}\label{thm:BADLdense}
$\BADL:\BA\to\DL$ is co-dense and the inserters \eqref{eq:DLinserter} form a co-density presentation in the sense of \cite[Thm 5.19]{1982:kelly}.
\end{theorem}

Together with \cite[Thm 5.30]{1982:kelly} we obtain

\begin{corollary}\label{cor:positivication}
Every $\BA$-endofunctor $L$ has a positivication $L'$ which can be computed explicitly at any distributive lattice $A$ as the inserter
\begin{equation}\label{eq:PositivicationIns}
\hfill
\xymatrix@C=10ex
{
L' A\ar[r] & \ForgBADL L \FreeBADL A\ar@<6pt>[r]^-{\ForgBADL L \FreeBADL in_2} \ar@<-6pt>[r]_-{\ForgBADL L \FreeBADL in_1}^-{\uparrow} & \ForgBADL L \FreeBADL (\two\cop A )
}
\hfill
\end{equation}
Moreover, a functor is a positivication iff it preserves Boolean algebras and the inserters \eqref{eq:DLinserter}.
\end{corollary}


\begin{proposition}\label{prop:finitarypositivication}
If $L:\BA\to\BA$ is finitary, then so is its positivication $L':\DL\to\DL$.
\end{proposition}
\begin{proof}
All  operations involved in \eqref{eq:PositivicationIns}, that is, $\BADL$, $L$, $\FreeDLBA$, $\two\bullet-$, preserve filtered colimits. And filtered colimits commute with finite weighted limits, see \cite[Prop.4.9]{Kelly:finitelimits}.\footnote{We are grateful to John Power for pointing out this reference.}
\end{proof}

\subsection{Positivication of normal modal logic}

Of course, the positivication of Kripke's normal modal logic with one meet-preserving $\Box$ should turn out to be Dunn's positive modal logic \cite{1995:DunnPML}. We show this in a roundabout way which has the advantage of making precise the relationship of our notion of positivication with the procedure employed in \cite{2013:PositiveFragments}.

To summarise, going back to Diagrams \eqref{diag:boolean} and \eqref{diag:positive}, \cite{2013:PositiveFragments} starts with $T$ and then, on the one hand define $L$ via $LB=\Pow T\Uf B$ on finite $\BA$s and, on the other hand, define $L'$ via $L'A=\Up T'\Pf A$ on finite $\DL$s with $T'$ the posetification of $T$. 

\begin{theorem}\label{thm:positivication}
Let $T:\Set\to\Set$ preserve finite sets and let $L$ be given on finite $B\in\BA$ by $LB=\Pow T\Uf B$. Let $T'$ be the posetification of $T$ and let $L'$ be given on finite $A\in\DL$ by $L'A=\Up T'\Pf A$. Then $L'$ is the positivication of $L$.
\end{theorem}

\begin{proof}
We have to show that $L'=\Ran_W WL$. By duality and definition of $T'$ as $\Lan_D DT$, we know that $L'$ and $\Ran_W WL$ agree on finite $\DL$s. Now the claim follows from Prop.\ref{prop:finitarypositivication}.
\end{proof}

\begin{rem}
\begin{enumerate}
\item 
The conditions of the theorem are not strong enough to guarantee that $L'$ is strongly finitary and thus has a presentation by operations and equations. As shown in \cite[Thm 6.20]{2013:PositiveFragments}, this is the case if $T$ preserves weak pullbacks.
\item In the case of graded modal logic, $L$ is different from $\Pow T\Uf$ even on finite $\BA$s. Therefore, the approach of \cite{2013:PositiveFragments} cannot be applied. We leave a description of the positivication of graded modal logic for a sequel.
\end{enumerate}
\end{rem}

\noindent
Now, if, in the notation of the theorem, we start with $T$ as the powerset functor, it is well known that $L$ is Kripke's normal modal logic and \cite{2013:PositiveFragments} shows that $L'$ is Dunn's positive modal logic. It follows from Theorem \ref{thm:positivication}, that the latter is indeed the positivication of the former.

\subsection{Positivication of non-monotone modal logics}

The basic idea of positivication is the following. Given a modal logic with monotone modalities, add the duals and find the axioms so that boolean negation can be pushed to the atoms. But our abstract definition of positivication is powerful enough to also apply to logics which have modalities that are not monotone. 

We will study what happens in such a situation through the example of the modal logic with one $\Box$ that does not obey any equations, not even monotonicity. In our functorial setting, this logic is given by $L=\FreeBA\ForgBA:\BA\to\BA$. Recall the functors $\BADL,\DLBA,\FreeBADL$ from \eqref{eq:GWK}. 

\begin{theorem}\label{thm:positiveFU}
The positivication of $\FreeBA\ForgBA:\BA\to\BA$ is 
$\BADL\FreeBA\ForgBA \DLBA:\DL\to\DL$.
\end{theorem}

\begin{proof}
We know from Theorem~\ref{thm:NbPoset} that $\Lan_\Disc \Disc\ForgBA\FreeBA = \Disc\ForgBA\FreeBA \Con$. By duality, on finite $\DL$s, $\Ran_\BADL \BADL\FreeBA\ForgBA = \BADL\FreeBA\ForgBA\DLBA$. Now the result follows from Prop.\ref{prop:finitarypositivication} since all of $\BADL,\FreeBA,\ForgBA,\DLBA$ are finitary.
\end{proof}

\begin{remark} From a logical point of view,  the appearance of $\DLBA:\DL\to\BA$ in Thm \ref{thm:positiveFU} tells us that, given $A\in \DL$, we are only allowed to build a formula $\Box a$, $a\in A$, \emph{if $a$ lies in a boolean subalgebra} (ie $a$ has a complement). This side condition takes us out of the realm of equational logic and, hence, of modal logics given  by axioms.  
This is related to the fact that $\DLBA$ is \emph{not} strongly finitary \cite[Example 6.6]{2013:PositiveFragments} and, therefore, functors involving $\DLBA$ cannot be expected to have a presentation by operations and equations. 
\end{remark}

To give another example of an extension by non-monotone modalities, the logic $\BADL_o\FreeBADL:\DL_o\to\DL_o$ is a modal logic over distributive lattices with one unary modality obeying the axioms of negation. In other words, $\BADL_o\FreeBADL$-algebras over $\DL_o$ are just boolean algebras. Clearly, $\BADL_o\FreeBADL$ is not locally monotone and negation, considered as a unary modality, cannot be `positivised'. Nevertheless, writing $I$ for the inclusion $\DL_o\to\DL$, the right Kan extension $\Ran_I I\BADL_o\FreeBADL$ does exist and is the identity.

\begin{proposition}
$\Ran_I I\BADL_o\FreeBADL = \Id$.
\end{proposition}

\begin{proof}
Going back to \eqref{eq:2stageKan}, we have $\tilde H=\BADL_o\FreeBADL$, which means that we can take $H=\Id$. But then, by Thm~\ref{thm:BADLdense} and \cite[Thm 5.1]{1982:kelly}, we have $H'=\Id$.
\end{proof}

To summarise, we have seen two examples of positivication of modal extensions by non-monotone modalities. In the first case, the non-monotone modality was made monotone by adding a side-condition restricting its use. In the second case, the non-monotone modality was eliminated.

\section{Positive coalgebraic logic}\label{sec:PCL}
\newcommand{\lang}{\mathcal{L}}
\subsection{Semantics}

Recall from the introduction that we wish to move from an ordinary BCL given by the diagram (\ref{diag:boolean}) to a $\Pos$-enriched PCL given by diagram (\ref{diag:positive}). In Sections \ref{sec:Posetification} and \ref{sec:Positivication} we have shown how to build $T'$ from $T$ and $L'$ from $L$ respectively. The missing element is the construction of $\delta'$ from $\delta$. Let us first remind the reader of how $\delta$ defines the interpretation, this will also be the occasion to fix some notation.
\begin{theorem}[\cite{2012:kurz-rosicky}]
An endofunctor $L$ on a variety $\Var$ has a finitary presentation by operations and
equations iff it preserves sifted colimits, in which case $\Alg(L)$ is a variety.
\end{theorem}
For $L$ strongly finitary on either $\BA$ or $\DL$ and $A$ an object of the corresponding category, let $\Free_L(A)$ denote the free $L$-algebra over $A$. In particular, if $V$ denotes a countable set of propositional variables and $\Free V$ is the freely generated object over $V$ in $\BA$ or $\DL$, then the free $L$-algebra $\Free_L(\Free V)$ is the algebra of $L$-modal formulas for the syntax defined by $L$, which we denote more succinctly by $\lang$. Now, let $\gamma: X\to TX$ be a $T$-coalgebra and assume that it comes equipped with a \emph{valuation} $v:\Free V\to \Pow X$, the interpretation map $\lsem -\rsem_{(\gamma,v)}$ is the unique map given by initiality of $\lang$ amongst $L(-)+\Free V$-algebras.
\[
\xymatrix@C=12ex@R=3ex
{
L\lang+\Free V\ar[rr]_{\simeq} \ar[d]_{L\lsem-\rsem_{(\gamma,v)}+\id_{\Free V}} & & \lang\ar[d]^{\lsem-\rsem_{(\gamma,v)}}\\
L\Pow X +\Free V\ar[r]^{\delta_X+\id_{\Free V}} & \Pow T X+\Free V\ar[r]^{\Pow\gamma+v}  &\Pow X
}
\]

We can now turn to defining $\delta'$ from $\delta$. To avoid unsightly $(-)\op$ symbols appearing everywhere we simply consider $\Pow,\Up,\Uf,\Pf$ to be \emph{contravariant functors} throughout (as opposed to covariant functors from/to an $(-)\op$ category). The following definition was given in \cite{2013:PositiveFragments}.
\begin{definition}
A logic $(L',\delta')$ for $T'$ is \emph{a positive fragment} of the logic $(L,\delta)$ for $T$, if there exist natural transformations $\alpha:T'\Disc\to \Disc T$ and $\beta: L'\ForgBADL \to \ForgBADL L$ such that $\ForgBADL\delta\circ \beta \Uf = \Pf\alpha \circ \delta' \Disc$.
\end{definition} 

Clearly, we have natural transformations $\alpha: T'\Disc\to \Disc T$ and $\beta: L'\ForgBADL \to \ForgBADL L$ by construction of the posetification $T'$ and of the positivication $L'$. We can construct a natural transformation $\delta'$ as follows. First, it is not hard to check that $\ForgBADL\circ \Pow=\Up\circ\Disc$. Thus, given a natural transformation $\delta: L\Pow \to \Pow T$ we get a natural transformation
\[
L'\Up \Disc  =L'\ForgBADL\Pow\stackrel{\beta_{\Pow}}{\Longrightarrow} \ForgBADL L\Pow\stackrel{W\delta}{\Longrightarrow}W\Pow T = \Up \Disc T 
\]

\begin{lemma}
For any poset $X$, the following diagram is an inserter:
\[
\xymatrix
{
\Up T'X\ar[r] & \Up \Disc T\ForgPos X \ar@<3pt>[r]^-{\Up\Disc T\pi_1}\ar@<-3pt>[r]_-{\Up\Disc T\pi_0} & \Up \Disc T\ForgPos X^\two 
}
\]
\end{lemma}
\begin{proof}
$\Disc TX_0\rightrightarrows \Disc TX_1\to T'X$ is a coinserter and since $\Up$ is the enriched hom functor $\hom(-,\two)$ it turns coinserters into inserters.
\end{proof}


By naturality of $\beta$ and $\delta$, the two right-hand side squares of the following diagram commute, and this defines a $\Pos$-enriched natural transformation $\delta': L'\Up\to \Up (T')$.
\begin{equation}\label{diag:deltaprime}
\xymatrix@R=5ex
{
L'\Up X \ar@{-->}[d]_{\delta'_X}\ar[r]& L'\Up \Disc \ForgPos X \ar@<3pt>[r]^-{L'\Up\Disc \pi_1}\ar@<-3pt>[r]_-{L'\Up\Disc \pi_0}\ar[d]_{\ForgBADL \delta_{\ForgPos X}\circ \beta_{\Pow \ForgPos X}} & L'\Up \Disc \ForgPos X^\two \ar[d]^{\ForgBADL\delta_{\ForgPos X^\two }\circ \beta_{\Pow \ForgPos X^\two }}\\
\Up T'X\ar[r] & \Up \Disc T\ForgPos X\ar@<3pt>[r]^-{\Up\Disc T\pi_1}\ar@<-3pt>[r]_-{\Up\Disc T\pi_0} & \Up \Disc T\ForgPos X^\two 
}
\end{equation}

\begin{theorem}
With $\delta'$ defined as above, the logic $(L',\delta')$ for $T'$ is a positive fragment of the logic $(L,\delta)$ for $L'$. 
\end{theorem}
\begin{proof}
We need to check that $\ForgBADL\delta\circ \beta \Uf = \Pf\alpha \circ \delta' \Disc$. Given a set $X$, the poset $\Disc X$ has a completely trivial coinserter presentation given by $\Disc X \rightrightarrows \Disc X \to \Disc X$, and in particular $T'\Disc X=\Disc TX$, i.e. $\alpha_X=\id_X$, and the result follows from the diagram (\ref{diag:deltaprime}).
\end{proof}
 

\subsection{Completeness}
We say that a BCL or a PCL $(L,\delta)$ is \emph{weakly complete} for $T$-coalgebras if for any formulas $\phi,\psi\in\lang$ such that $\phi \nleq\psi$, there exists a $T$-coalgebra $\gamma:X\to TX$ and a valuation $v: \FreeBA V\to\Pow X$ ($v: \FreeDL V\to \Up X$ for posets) and an element $x\in X$ such that $
 x\in\lsem \phi\rsem_{(\gamma,v)}$  but $x\notin\lsem \psi\rsem_{(\gamma,v)}$.
The following theorem gives a sufficient condition for weak completeness.

\begin{theorem}[\cite{2004:KKP}]
A BCL or a PCL $(L,\delta)$ for a functor $T$ is weakly complete if $\delta$ is component-wise injective.
\end{theorem}

Weak completeness transfers from a boolean logic to its positive fragment. 

\begin{theorem}\label{thm:weakCompTransfer}
For a BCL $(L,\delta)$ defined by a strongly finitary functor $L:\BA\to\BA$, if $\delta$ is component-wise injective, then so is $\delta'$. In particular $(L',\delta')$ is then weakly complete.
\end{theorem}
\begin{proof}
Recall first that finitary (and thus strongly finitary) functors $L:\BA\to\BA$ preserve injective maps (\cite[Lemma 6.14]{2010:AlexanderDaniela}). Since the natural transformation $\beta: L'\ForgBADL \to\ForgBADL L$ is an isomorphism, it follows that the vertical legs of (\ref{diag:deltaprime}) are injective. Since $\Up=(-,\two)$ turns coinserters into inserters and $L'$ preserves inserters by construction, the top row of (\ref{diag:deltaprime}) is an inserter as well, and hence injective. It follows that $\delta'_X$ must be injective.
\end{proof}

\noindent \textbf{The case of normal modal logic.} Let $L:\BA\to\BA$ be the syntax-building functor for normal modal logic:
\[
LA=(\FreeBA\ForgBA A)/ \{\dia(a\vee b)=\dia a\vee \dia b, \dia\bot=\bot\}
\]
and let $\delta: L\Pow\to\Pow\cPow$ be the semantic transformation for normal modal logic, i.e.
\[
\delta_X(\dia U)=\{V\subseteq X\mid V\cap U\neq\emptyset\}
\]
We have computed the posetification $\Pow'$ or $\Pow$ in Section \ref{sec:Posetification} and shown that the positivication $L'$ of $L$ is given by Dunn's syntax (\cite{1995:DunnPML}) in Section \ref{sec:Positivication}. The following result is well-known, and can be shown directly.
\begin{theorem}\label{thm:deltaMLmono}
The natural transformation $\delta: L\Pow\to\Pow\cPow$ is component-wise injective.
\end{theorem}

\begin{corollary}
If $L'$ is the positivication of the syntax functor for normal modal logic, $\cPow'$ the posetification of the powerset functor, and $\delta': L'\Up\to\Up\cPow'$ the semantics generated from $\delta: L\Pow\to\Pow \cPow$ by diagram (\ref{diag:deltaprime}), then the PCL $(L',\delta')$ is weakly complete for $\cPow'$-coalgebras.
\end{corollary}

\section{Conclusion and future work}
We have presented positive coalgebraic logic at the same level of generality as boolean coalgebraic logic, and developed a method by which boolean coalgebraic logics can systematically be turned into positive coalgebraic logics. We have also shown that completeness follows automatically from the boolean case in this setup. More broadly,
we have also presented a practical application of enriched category theory in logic by showing that positive modal logic amounts to a type of $\Pos$-enriched coalgebraic logic. We believe that this perspective offers a deep insight into the fundamental difference between boolean and positive logics.

Much remains to be investigated. First, we do not yet have much practical experience and tools to compute positivications. As Section \ref{sec:Positivication} illustrates, our calculations are all indirect. In particular we would like to compute the positivication of graded modal logic. On the logic side, we have good reasons to believe that strong completeness transfers from the boolean to the positive case for a large class of functors. On the other hand, we believe that expressivity does not transfer in general. All this will be investigated in a future companion publication to this work. 


\bibliography{PCL}

\end{document}